\documentclass{article}

\usepackage{microtype}
\usepackage{graphicx}
\usepackage{subfigure}
\usepackage{booktabs} 

\usepackage{hyperref}

\usepackage[accepted]{icml2023}

\usepackage{microtype}
\usepackage{inconsolata}
\usepackage{amsmath}
\usepackage{amssymb}
\usepackage{mathtools}
\usepackage{amsthm}
\usepackage{adjustbox}
\usepackage{booktabs}
\usepackage{multirow}

\usepackage{paralist}
\usepackage{enumitem}

\usepackage{mdwlist}
\usepackage{graphicx}

\usepackage{color}

\usepackage{xspace}

\usepackage[capitalize,noabbrev]{cleveref}

\theoremstyle{plain}
\newtheorem{theorem}{Theorem}[section]

\newtheorem{lemma}[theorem]{Lemma}

\theoremstyle{definition}

\theoremstyle{remark}

\usepackage[textsize=tiny]{todonotes}

\icmltitlerunning{Protecting Language Generation Models via Invisible Watermarking}

\newcommand{\method}{\textsc{Ginsew}\xspace}

\begin{document}

\twocolumn[
\icmltitle{Protecting Language Generation Models via Invisible Watermarking}
\icmlsetsymbol{equal}{*}
\begin{icmlauthorlist}
\icmlauthor{Xuandong Zhao}{yyy}
\icmlauthor{Yu-Xiang Wang}{yyy}
\icmlauthor{Lei Li}{yyy}
\end{icmlauthorlist}

\icmlaffiliation{yyy}{Department of Computer Science, UC Santa Barbara}

\icmlcorrespondingauthor{Xuandong Zhao}{xuandongzhao@cs.ucsb.edu}
\icmlcorrespondingauthor{Yu-Xiang Wang}{yuxiangw@cs.ucsb.edu}
\icmlcorrespondingauthor{Lei Li}{leili@cs.ucsb.edu}


\icmlkeywords{Machine Learning, ICML}
\vskip 0.3in
]



\printAffiliationsAndNotice{}  

\begin{abstract}

Language generation models have been an increasingly powerful enabler for many applications. Many such models offer free or affordable API access, which makes them potentially vulnerable to model extraction attacks through distillation. To protect intellectual property (IP) and ensure fair use of these models, various techniques such as lexical watermarking and synonym replacement have been proposed. However, these methods can be nullified by obvious countermeasures such as ``synonym randomization''. To address this issue, we propose \method, a novel method to protect text generation models from being stolen through distillation. The key idea of our method is to inject secret signals into the probability vector of the decoding steps for each target token. We can then detect the secret message by probing a suspect model to tell if it is distilled from the protected one. Experimental results show that \method can effectively identify instances of IP infringement with minimal impact on the generation quality of protected APIs. Our method demonstrates an absolute improvement of 19 to 29 points on mean average precision (mAP) in detecting suspects compared to previous methods against watermark removal attacks.
\end{abstract}

\section{Introduction}
\label{sec:intro}

Large language models (LLMs) have become increasingly powerful \cite{Brown2020LanguageMA, Ouyang2022TrainingLM}, but their owners are reluctant to open-source them due to high training costs. Most companies provide only API access to their models for free or for a fee to cover innovation and maintenance costs. While many applications benefit from these APIs, some are looking for cheaper alternatives. 

While healthy competition is good for preventing monopoly in the LLM service market, a fairly priced API-service may expose a ``short cut'' that allows a company to distill a comparable model at a much lower cost --- forcing the original model creator out of the market and stifling future innovation. Anticipating this, rational LLM owners will likely never provide API access at a fair price. Instead, they must dramatically increase fees - far beyond service costs - to make model distillation unprofitable. 
\begin{figure}[t]
    \centering
    \includegraphics[width=1.0\linewidth]{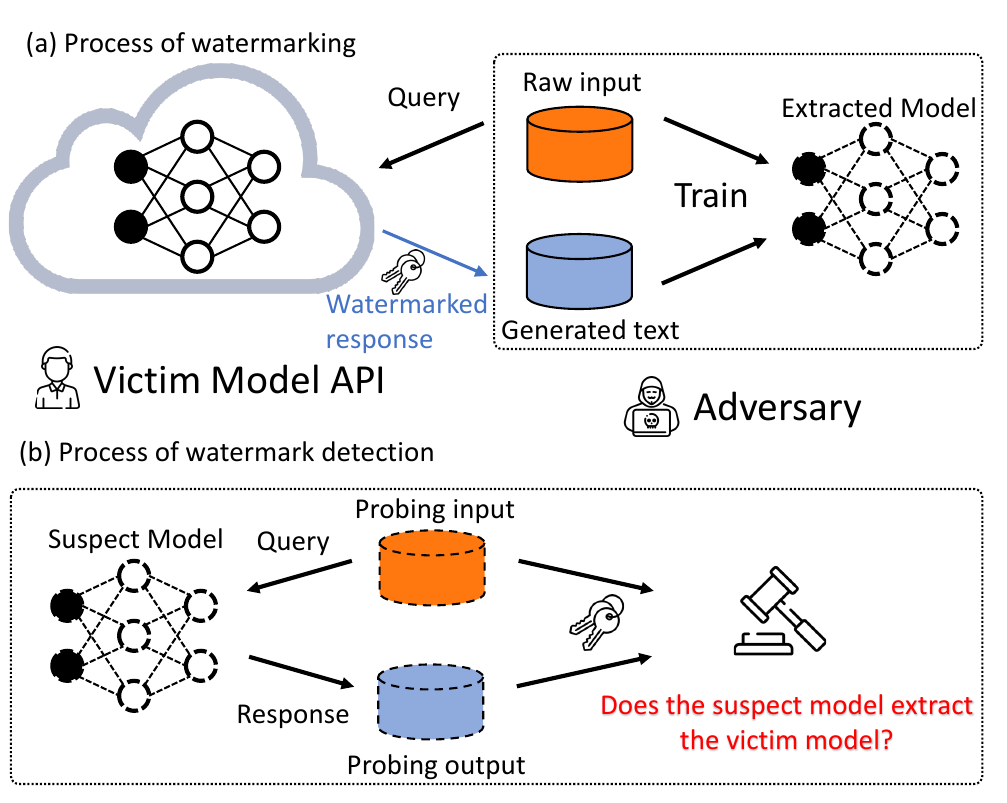}
    \caption{Overview of the process of watermarking and the process of watermark detection. The victim model API embeds watermarks in the response to input queries from the adversaries. The API owner can then use a key to verify if the suspect model has been distilled from the victim model.}
    \label{fig:overview}
\end{figure}

The increased fee might be tolerable for downstream applications with good commercial values, but it makes research and LLM applications for public good unaffordable. To address this, governments could fund public LLMs for research and public service. However, public LLMs could also be distilled for commercial gain or military use by unethical entities.

Recently, a Stanford group \cite{alpaca} claimed to have distilled ChatGPT for just \$600 in API calls, demonstrating the issue’s urgency. The cheaper cost to distill versus retrain models can be justified theoretically \cite{shalev2014understanding}: The sample complexity for training a model from raw data to $\epsilon$-excess risk is on the order of $O(1/\epsilon^2)$; whereas model distillation that uses model-generated labels operates in a ``realizable'' regime, thus requires only $O(1/\epsilon)$ samples.

This brings us to the central question: How can we prevent model-stealing attacks through distillation?

It appears to be a mission impossible. By providing the state-of-the-art LLM service, the information needed to replicate it has been encoded in the provided API outputs themselves. Researchers decided to give up on preventing learnability, but instead, try to watermark the API outputs so models trained using the watermarked outputs can be traced to the original model. Recent works \cite{adi2018turning, zhang2018protecting} use a trigger set to embed invisible watermarks on the neurons of the commercial models before distribution. When model theft is suspected, model owners can conduct an official ownership claim with the aid of the trigger set. Existing works \cite{He2021ProtectingIP, He2022CATERIP} use a synonym replacement strategy to add surface-level watermarks to a text generation model.
However, these methods can be bypassed by an adversary randomly replacing synonyms in the output and removing watermarks from the extracted model, making the protection ineffective.

In this paper, we propose \underline{g}enerative \underline{in}visible \underline{se}quence \underline{w}atermarking (\method), a method to protect text generation models and detect stolen ones. Figure \ref{fig:overview} illustrates the overall process. The core idea is to inject a secret sinusoidal signal into the model's generation probabilities for words. This signal does not harm the model's generation quality. To detect whether a candidate model is stealing from the target model, \method identifies the sinusoidal frequency and compares it with the secret key. \method provides a more robust mechanism for protecting the intellectual property of the model, even under random synonym replacement attacks.

The contributions of this paper are as follows:
\begin{itemize}[itemsep=0pt,leftmargin=1em, topsep=1pt]
    \item We propose \underline{g}enerative \underline{in}visible \underline{se}quence \underline{w}atermarking (\method), a method to protect text generation models against model extraction attacks with invisible watermarks.
    \item We carry out experiments on machine translation and story generation on a variety of models. Experimental results show that our method \method outperforms the previous methods in both generation quality and robustness of IP infringement detection ability. Even with adversarial watermark removal attacks, \method still gains a significant improvement of 19 to 29 points in mean average precision (mAP) of detecting suspects\footnote{Our source code is available at \url{https://github.com/XuandongZhao/Ginsew}.}.
\end{itemize}

\section{Related work}
\label{sec:related}

\paragraph{Model extraction attacks}
Model extraction attacks, also known as model inversion or model stealing, pose a significant threat to the confidentiality of machine learning models \cite{tramer2016stealing, Orekondy2018KnockoffNS, Wallace2020ImitationAA, He2021model}. These attacks aim to imitate the functionality of a black-box victim model by creating or collecting a substitute dataset. The attacker then uses the victim model's APIs to predict labels for the substitute dataset. With this pseudo-labeled dataset,  the attacker can train a high-performance model that mimics the victim model and can even mount it as a cloud service at a lower price. This type of attack is closely related to knowledge distillation (KD) \cite{Hinton2015DistillingTK}, where the attacker acts as the student model that approximates the behavior of the victim model. In this paper, we specifically focus on model extraction attacks in text generation. Previous works \cite{Wallace2020ImitationAA, Xu2021StudentST} have shown that adversaries can use sequence-level knowledge distillation to mimic the functionality of commercial text generation APIs, which poses a severe threat to cloud platforms.
\paragraph{Watermarking}
Watermarking is a technique used to embed unseen labels into signals such as audio, video, or images to identify the owner of the signal's copyright. In the context of machine learning models, some studies \cite{Merrer2017AdversarialFS, adi2018turning, zhang2018protecting} have used watermarks to prevent exact duplication of these models, by inserting them into the parameters of the protected model or constructing backdoor images that activate specific predictions. However, protecting models from model extraction attacks is difficult as the parameters of the suspect model may be different from those of the victim model and the backdoor behavior may not be transferred. 

\paragraph{Watermarking against model extraction}
Recently, several studies have attempted to address the challenge of identifying extracted models that have been distilled from a victim model, with promising results in image classification \cite{Charette2022CosineMW, Jia2020EntangledWA} and text classification \cite{zhao2022distillation}. CosWM \cite{Charette2022CosineMW} embeds a watermark in the form of a cosine signal into the output of the protected model. This watermark is difficult to eliminate, making it an effective way to identify models that have been distilled from the protected model. Nonetheless, CosWM only applies to image classification tasks. As for text generation, \citet{He2021ProtectingIP} propose a lexical watermarking method to identify IP infringement caused by extraction attacks. This method involves selecting a set of words from the training data of the victim model, finding semantically equivalent substitutions for them, and replacing them with the substitutions. Another approach, CATER \cite{He2022CATERIP}, proposes conditional watermarking by replacing synonyms of some words based on linguistic features. However, both methods are surface-level watermarks. The adversary can easily bypass these methods by randomly replacing synonyms in the output, making it difficult to verify by probing the suspect models. \method directly modifies the probability distribution of the output tokens, which makes the watermark invisible and provides a more robust mechanism for identifying extracted models.



\section{Proposed method: \method}
\label{sec:approach}

\begin{algorithm}[t]
   \caption{Watermarking process}
   \label{alg:wm}
\begin{algorithmic}[1]
   \STATE {\bfseries Inputs:} Input text $\boldsymbol{x}$, probability vector $\mathbf{p}$ from the decoder of the victim model, vocab $\mathcal{V}$, group 1 $\mathcal{G}_1$, group 2 $\mathcal{G}_2$, hash function $g(\boldsymbol{x}, \mathbf{v}, \mathbf{M})$. \\
   \STATE {\bfseries Output:} Modified probability vector $\mathbf{p}$
   \STATE Calculate probability summation of tokens in group 1 and group 2:
   $Q_{\mathcal{G}_1} =\sum_{i \in \mathcal{G}_1} \mathbf{p}_i, ~ Q_{\mathcal{G}_2} = \sum_{i \in \mathcal{G}_2} \mathbf{p}_i$
   \STATE Calculate the periodic signal 
   \begin{align*}
       z_{1}(\boldsymbol{x}) & = \cos \left(f_{w} g(\boldsymbol{x}, \mathbf{v}, \mathbf{M})\right), \\
       z_{2}(\boldsymbol{x}) & = \cos \left(f_{w} g(\boldsymbol{x}, \mathbf{v}, \mathbf{M}) + \pi\right)  
   \end{align*}
   \STATE Set $\tilde{Q}_{\mathcal{G}_1} = \frac{Q_{\mathcal{G}_1}+\varepsilon\left(1+z_{1}(\boldsymbol{x})\right)}{1+2 \varepsilon}, \tilde{Q}_{\mathcal{G}_2} = \frac{Q_{\mathcal{G}_2}+\varepsilon\left(1+z_{2}(\boldsymbol{x})\right)}{1+2 \varepsilon}$
   \FOR{$i=1$ {\bfseries to} $|\mathcal{V}|$}
   \STATE{\bfseries if }{$i \in \mathcal{G}_1$ }{\bfseries then }$\mathbf{p}_i \leftarrow \frac{\tilde{Q}_{\mathcal{G}_1}}{Q_{\mathcal{G}_1}}\cdot \mathbf{p}_i$
    \STATE{\bfseries else } $\mathbf{p}_i \leftarrow \frac{\tilde{Q}_{\mathcal{G}_2}}{Q_{\mathcal{G}_2}}\cdot \mathbf{p}_i$
   \ENDFOR
   \STATE {\bfseries return} $\mathbf{p}$
\end{algorithmic}
\end{algorithm}

\subsection{Problem setup}
Our goal is to protect text generation models against model extraction attacks. It enables the victim model owner or a third-party arbitrator to attribute the ownership of a suspect model that is distilled from the victim API. We leverage the secret knowledge that the extracted model learned from the victim model as a signature for attributing ownership. 

In a model extraction attack, the adversary, denoted as $\mathcal{S}$, only has black-box access to the victim model's API, denoted as $\mathcal{V}$. The adversary can query $\mathcal{V}$ using an auxiliary unlabeled dataset, but can only observe the text output and not the underlying probabilities produced by the API. As a result, the adversary receives the output of $\mathcal{V}$ as generation output or pseudo labels. The attacker's goal is to employ sequence-level knowledge distillation to replicate the functionality of $\mathcal{V}$. The model extraction attack is depicted in Figure \ref{fig:overview}(a).

Our method does not aim to prevent model extraction attacks as we cannot prohibit the adversary from mimicking the behavior of a common user. Instead, we focus on verifying whether a suspect model has been trained from the output of the victim model API. If a suspect model distills well, it carries the signature inherited from the victim. Otherwise, a suspect may not carry such a signature, and its generation quality will be noticeably inferior to the victim.

During verification, we assume the presence of a third-party arbitrator (e.g., law enforcement) with white-box access to both the victim and suspect models, as well as a probing dataset. The arbitrator or model owner compares the output of the suspect model to the secret signal in the watermark with a key: if the output matches the watermark, the suspect model is verified as a stolen model.

\subsection{Generative invisible sequence watermarking}
\method dynamically injects a watermark in response to queries made by an API’s end-user. It is invisible because it is not added to the surface text. We provide an overview of \method in Figure \ref{fig:overview}, consisting of two stages: i) watermarking stage and ii) watermark detection stage.


\method protects text generation models from model extraction attacks. We achieve this by carefully manipulating the probability distribution of each token generated by the model, specifically by modifying the probability vector during the decoding process. This approach allows for a unique signature to be embedded within the generated text, making it easily identifiable by the owner or a third-party arbitrator, while still maintaining the coherence and fluency of the text. The process of \method is illustrated in Figure \ref{fig:idea}.

\begin{figure}[t]
    \centering
    \includegraphics[width=1.0\linewidth]{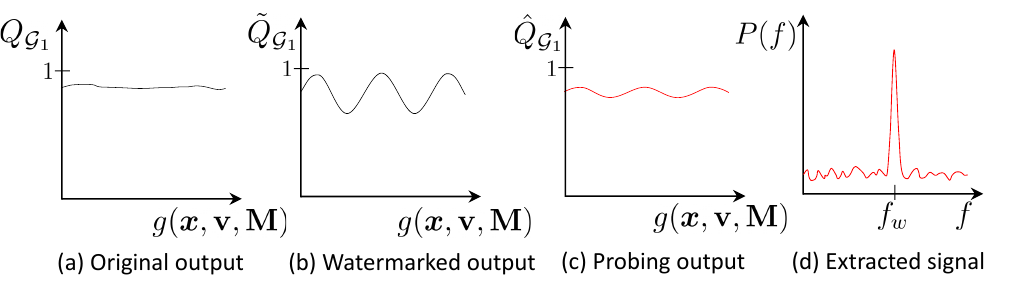}
    \caption{The process of \method. (a) The original group probability of the victim model is represented by  $Q_{\mathcal{G}_1}$. (b) The API owner applies a sinusoidal perturbation to the predicted group probability, resulting in a watermarked output, denoted as  $\tilde{Q}_{\mathcal{G}_1}$. (c) If the adversary attempts to distill the victim model, the extracted model will convey this periodical signal. (d) After applying a Fourier transform to the output with a specific key, a peak in the frequency domain at frequency $f_w$ can be observed.}
    \label{fig:idea}
\end{figure}

We present the watermarking process in Algorithm \ref{alg:wm}. For each token $v$ in the whole vocabulary $\mathcal{V}$, we randomly assign it to two distinct groups, group 1 $\mathcal{G}_1$ and group 2 $\mathcal{G}_2$, so that each group contains $\frac{|\mathcal{V}|}{2}$ words. $|\mathcal{V}|$ represents the vocabulary size. For each input text $\boldsymbol{x}$, we use a hash function, referred to as $g(\cdot)$, to project it into a scalar.
This hash function takes three inputs: input text $\boldsymbol{x}$, phase vector $\mathbf{v} \in \mathbb{R}^n$ and token matrix $\mathbf{M}\in \mathbb{R}^{|\mathcal{V}|\times n}$. Note that each token corresponds to a row in the matrix. The elements of the phase vector $\mathbf{v}$ are randomly sampled from a uniform distribution $[0, 1)$, while the elements of the token matrix $\mathbf{M}$ are randomly sampled from a standard normal distribution $\mathbf{M}_{ij} \sim \mathcal{N}(0, 1)$. Let $\mathbf{M}_{i} \in \mathbb{R}^{n}$ denote the $i$-th row of matrix $\mathbf{M}$, $\mathbf{v}^\top\mathbf{M}_{i} \sim \mathcal{N}(0, \frac{n}{3})$ (See proof in Appendix \ref{lemma1}). We then use the probability integral transformation $F$ to obtain a uniform distribution of the hash values: $g(\boldsymbol{x}, \mathbf{v}, \mathbf{M}) = F(\mathbf{v}^\top\mathbf{M}_{\text{tok}(\boldsymbol{x})}) \sim \mathcal{U}(0, 1)$, where $\text{tok}(\boldsymbol{x})$ denotes the ID of the second token in the input text. 

We design a periodic signal function based on the input. The hash function $g(\boldsymbol{x}, \mathbf{v}, \mathbf{M})$ returns a scalar determining the amount of noise to be added to the output.
As shown in Algorithm \ref{alg:wm}, given a probability vector $\mathbf{p} \in [0,1]^{|\mathcal{V}|}$, we calculate the total probability of group 1, $Q_{\mathcal{G}_1} =\sum_{i \in \mathcal{G}_1} \mathbf{p}_i$, and the total probability of group 2, $Q_{\mathcal{G}_2} = \sum_{i \in \mathcal{G}_2} \mathbf{p}_i$.
Following \citet{Charette2022CosineMW}, we calculate the periodic signal function, where $f_w \in \mathbb{R}$ is the angular frequency.
   \begin{align}
       z_{1}(\boldsymbol{x}) & = \cos \left(f_{w} g(\boldsymbol{x}, \mathbf{v}, \mathbf{M})\right), \\
       z_{2}(\boldsymbol{x}) & = \cos \left(f_{w} g(\boldsymbol{x}, \mathbf{v}, \mathbf{M})+ \pi \right)  
   \end{align}
We can obtain a watermark key by combining a set of variables $K = (f_w, \mathbf{v}, \mathbf{M})$. Note that $z_{1}(\boldsymbol{x}) + z_{2}(\boldsymbol{x}) = 0$. Next, we compute the periodic signal for the modified group probability 
\begin{align}
    \tilde{Q}_{\mathcal{G}_1} &= \frac{Q_{\mathcal{G}_1}+\varepsilon\left(1+z_{1}(\boldsymbol{x})\right)}{1+2 \varepsilon}, \label{eq:mod1} \\
    \tilde{Q}_{\mathcal{G}_2} &= \frac{Q_{\mathcal{G}_2}+\varepsilon\left(1+z_{2}(\boldsymbol{x})\right)}{1+2 \varepsilon} \label{eq:mod2}
\end{align} We prove that $0 \leq Q_{\mathcal{G}_1},Q_{\mathcal{G}_2} \leq 1 $ and $Q_{\mathcal{G}_1}+Q_{\mathcal{G}_2} = 1$ (In Appendix \ref{lemma2}). $\varepsilon$ is the watermark level, measuring how much noise is added to the group probability. Then we change the elements in the probability vector such that they align with the perturbed group probability. In this way, we are able to embed a hidden sinusoidal signal into the probability vector of the victim model. We can use decoding methods like beam search, top-$k$ sampling, or other methods to generate the output with the new probability vector.

\begin{algorithm}[t]
   \caption{Watermark detection}
   \label{alg:ext_wm}
\begin{algorithmic}[1]
   \STATE {\bfseries Inputs:} Suspect model $\mathcal{S}$, sample probing data $\mathcal{D}$ from the training data of $\mathcal{S}$, vocab $\mathcal{V}$, group 1 $\mathcal{G}_1$, group 2 $\mathcal{G}_2$, hash function $g(\boldsymbol{x}, \mathbf{v}, \mathbf{M})$, filtering threshold value $q_{\min}$. \\
   \STATE {\bfseries Output:} Signal strength
   \STATE Initialize $\mathcal{H} = \emptyset$
   \FOR{{\bfseries each} input $\boldsymbol{x} \text{ in } \mathcal{D}$ }
   \STATE $t = g(\mathbf{v}, \boldsymbol{x}, \mathbf{M})$
   \FOR{{\bfseries each} decoding step of $\mathcal{S}(\boldsymbol{x})$}
   \STATE Get probability vector $\hat{\mathbf{p}}$ from the decoder of the suspect model.
   \STATE $\hat{Q}_{\mathcal{G}_1} =\sum_{i \in \mathcal{G}_1}\hat{\mathbf{p}}_i$
   \STATE $\mathcal{H} \leftarrow \mathcal{H} \cup (t,\hat{Q}_{\mathcal{G}_1}) $
   \ENDFOR
   \ENDFOR
   \STATE Filter out elements in $\mathcal{H}$ where $\hat{Q}_{\mathcal{G}_1}\leq q_{\min}$, remaining pairs form the set $\widetilde{\mathcal{H}}$.
   \STATE Compute the Lomb-Scargle periodogram from the pairs $(t^{(k)},{\hat{Q}_{\mathcal{G}_1}}^{(k)}) \in \widetilde{\mathcal{H}}$
    \STATE Compute $P_{\text {snr}}$ in Equation \ref{eq:psnr}.
   \STATE {\bfseries return} $P_{\text {snr}}$
\end{algorithmic}
\end{algorithm}

\begin{algorithm}[htbp]
   \caption{Watermark detection with text alone}
   \label{alg:ext_wm2}
\begin{algorithmic}[1]
   \STATE {\bfseries Inputs:} Suspect model $\mathcal{S}$, sample probing data $\mathcal{D}$ from the training data of $\mathcal{S}$, vocab $\mathcal{V}$, group 1 $\mathcal{G}_1$, group 2 $\mathcal{G}_2$, hash function $g(\boldsymbol{x}, \mathbf{v}, \mathbf{M})$.\\
   \STATE {\bfseries Output:} Signal strength
   \STATE Initialize $\mathcal{H} = \emptyset$
   \FOR{{\bfseries each} input $\boldsymbol{x} \text{ in } \mathcal{D}$ }
   \STATE $t = g(\mathbf{v}, \boldsymbol{x}, \mathbf{M})$
   \STATE $\boldsymbol{y} \leftarrow \mathcal{S}(\boldsymbol{x}) $
   \FOR{{\bfseries each} token of $\boldsymbol{y}$}
   \STATE $\mathcal{H} \leftarrow \mathcal{H} \cup \left(t, \mathbf{1}(\boldsymbol{y}_i \in \mathcal{G}_1)\right) $
   \ENDFOR
   \ENDFOR
   \STATE Compute the Lomb-Scargle periodogram from $\mathcal{H}$, and compute $P_{\text {snr}}$ in Equation \ref{eq:psnr}.
   \STATE {\bfseries return} $P_{\text {snr}}$
\end{algorithmic}
\end{algorithm}

\subsection{Detecting watermark from suspect models}
In order to detect instances of IP infringement, we first create a probing dataset $\mathcal{D}$ to extract watermarked signals in the probability vector of the decoding steps. $\mathcal{D}$ can be obtained from the training data of the extracted model, as the owner has the ability to store any query sent by a specific end-user. Additionally, $\mathcal{D}$ can also be drawn from other distributions as needed.

The process of detecting the watermark from suspect models is outlined in Algorithm \ref{alg:ext_wm}. For each probing text input, the method first acquires the hash value $t = g(\boldsymbol{x}, \mathbf{v}, \mathbf{M})$. Next, for each decoding step of the suspect model, we add the pair $(t,\hat{Q}_{\mathcal{G}_1})$ to the set $\mathcal{H}$. To eliminate outputs with low confidence, we filter out the pairs with $\hat{Q}_{\mathcal{G}_1} \leq q_{\min}$, where the threshold value $q_{\min}$ is a constant parameter of the extraction process. The Lomb-Scargle periodogram \cite{Scargle1982StudiesIA} method is then used to estimate the Fourier power spectrum $P(f)$, at a specific frequency, $f_w$, in the probing set $\mathcal{H}$. By applying approximate Fourier transformation, we amplify the subtle perturbation in the probability vector. So that we can detect a peak in the power spectrum at the frequency $f_w$. This allows for the evaluation of the strength of the signal, by calculating the signal-to-noise ratio $P_{\text {snr}}$
\begin{gather}
P_{\text {signal }}=\frac{1}{\delta} \int_{f_{w}-\frac{\delta}{2}}^{f_{w}+\frac{\delta}{2}} P(f) df \notag \\ 
\text{\footnotesize $P_{\text {noise }}=\frac{1}{F-\delta}\left[\int_{0}^{f_{w}-\frac{\delta}{2}} P(f) d f+\int_{f_{w}+\frac{\delta}{2}}^{F} P(f) d f\right]$} \notag \\
P_{\text {snr}}=P_{\text {signal }} / P_{\text {noise }},
\label{eq:psnr}
\end{gather}
where $\delta$ controls the window width of $\left[f_{w}-\frac{\delta}{2}, f_{w}+\frac{\delta}{2}\right]$; $F$ is the maximum frequency, and $f_{w}$ is the angular frequency embedded into the victim model. A higher $P_{\text {snr}}$ indicates a higher peak in the frequency domain, and therefore a higher likelihood of the presence of the secret signal in the suspect model, confirming it distills the victim model. 

It is worth noting that the threshold and frequency parameters used in the Lomb-Scargle periodogram method can be adjusted to optimize the performance of the proposed method. The specific settings used in our experiments will be discussed in the Experiments section of the paper. Besides, we present Algorithm \ref{alg:ext_wm2} demonstrating that we can detect watermarks by analyzing just the generated text itself, without relying on the model's predicted probabilities. A more detailed discussion of this text-only approach can be found in Section \ref{sec:wm2}.

\section{Experiments}
\label{sec:exps}
\begin{table*}[t]
\centering

\setlength{\tabcolsep}{2pt} 
\caption{Main results for model performance and detection. We report the generation quality and mean average precision of the model infringement detection (Detect mAP $\times 100$). We use F1 scores of ROUGE-L and BERTScore.}
\resizebox{\textwidth}{!}{
\begin{tabular}{l|ccc|ccc|ccc}
\toprule
 & \multicolumn{3}{|c|}{\textbf{IWSLT14}}  & \multicolumn{3}{|c|}{\textbf{WMT14}} & \multicolumn{3}{c}{\textbf{ROCStories}} \\
 & BLEU $\uparrow$ & BERTScore $\uparrow$ & Detect mAP $\uparrow$ & BLEU $\uparrow$ & BERTScore $\uparrow$ & Detect mAP $\uparrow$ & ROUGE-L $\uparrow$ & BERTScore $\uparrow$ & Detect mAP $\uparrow$ \\
\midrule
Original models & 34.6 & 94.2 & - & 30.8 & 65.7 & - & 16.5 & 90.1 & - \\
\midrule
\multicolumn{10}{l}{Plain watermark} \\
~~~~\citet{He2021ProtectingIP} & 33.9 & 92.7 & 100 & 30.5 & 65.3 & 100 & 15.8 & 89.3 & 100 \\
~~~~CATER \cite{He2022CATERIP} & 33.8 & 92.5 & 76.4 & 30.5 & 65.4 & 78.3 & 15.6 & 89.1 & 83.2 \\
~~~~\method & 34.2 & 93.8 & 100 & 30.6 & 65.5 & 100 & 16.1 & 89.6 & 100 \\
\midrule
\multicolumn{10}{l}{Watermark removed by synonym randomization} \\
~~~~\citet{He2021ProtectingIP} & 32.7 & 90.7 & 63.1 & 29.6 & 64.7 & 62.3 & 14.8 & 88.4 & 59.6 \\
~~~~CATER \cite{He2022CATERIP} & 32.7 & 90.6 & 68.5 & 29.5 & 64.7 & 63.1 & 14.9 & 88.4 & 64.2 \\
~~~~\method & 33.1 & 90.9 & 87.7 & 29.8 & 64.9 & 86.9 & 15.1 & 89.0 & 93.2 \\
\bottomrule
\end{tabular}
}
\label{table:main}
\end{table*}

We evaluate the performance of \method on two common text generation tasks: machine translation and story generation. There are multiple public APIs available for these models\footnote{https://translate.google.com/}\footnote{https://beta.openai.com/overview}.

\paragraph{Machine translation}
In the machine translation task, we utilize the IWSLT14 and WMT14 datasets \cite{Cettolo2014ReportOT, Bojar2014FindingsOT}, specifically focusing on German (De) to English (En) translations. We evaluate the quality of the translations based on BLEU \cite{Papineni2002BleuAM} and BERTScore \cite{Zhang2019BERTScoreET} metrics. We adopt the official split of train/valid/test sets. BLEU focuses on lexical similarity by comparing n-grams, while BERTScore focuses on semantic equivalence through contextual embeddings. For IWSLT14, a vocabulary consisting of 7,000 BPE \cite{Sennrich2016NeuralMT} units is used, whereas WMT14 employs 32,000 BPE units.

\paragraph{Story generation}
For the story generation task, we use the ROCstories \cite{Mostafazadeh2016ACA} corpus. Each story in this dataset comprises 5 sentences, with the first 4 sentences serving as the context for the story and the input to the model, and the 5th sentence being the ending of the story to be predicted. There are 90,000 samples in the train set, and 4081 samples in the validation and test sets. The generation quality is evaluated based on ROUGE \cite{Lin2004ROUGEAP} and BERTScore metrics. A vocabulary of 25,000 BPE units is used in this task.

\paragraph{Baselines} 
We compare \method with \citet{He2021ProtectingIP} and CATER \cite{He2022CATERIP}. Specifically, \citet{He2021ProtectingIP} propose two watermarking approaches: the first one replaces all the watermarked words with their synonyms; the second one watermarks the victim API outputs by mixing American and British spelling systems. Because the second one is easily eliminated by the adversary through consistently using one spelling system, we focus on their first approach. This method selects a set of words $\mathcal{C}$ from the training data of the victim model. Then for each word $c\in\mathcal{C}$, it finds synonyms for $c$ and forms a set $\mathcal{R}$. Finally, the original words of $\mathcal{C}$ and their substitutions $\mathcal{R}$ are replaced with watermarking words $\mathcal{W}$. As an improvement of \citet{He2021ProtectingIP}, CATER proposes a conditional watermarking framework, which replaces the words with synonyms based on linguistic features. The hit ratio is used as the score for the baselines to detect IP infringement
\begin{equation}\label{eq:hit}
\text{hit}=\frac{\#\left(\mathcal{W}_y\right)}{\#\left(\mathcal{C}_y \cup \mathcal{R}_y\right)}
\end{equation}
where $\#\left(\mathcal{W}_y\right)$ represents the number of watermarked words $\mathcal{W}$ appearing in the suspect model's output $\boldsymbol{y}$, and $\#\left(\mathcal{C}_y \cup \mathcal{R}_y\right)$ is the total number of $\mathcal{C}_y \text{ and } \mathcal{R}_y$ found in word sequence $\boldsymbol{y}$. We reproduce the watermarking methods with the synonym size of 2 in \citet{He2021ProtectingIP} and CATER. For CATER, we use the first-order Part-of-Speech (POS) as the linguistic feature.

\paragraph{Evaluation} 
In order to evaluate the performance of \method against the baselines in detecting extracted models, we reduce the binary classification problem into a thresholding task by using a specific test score. Since \method and baselines use different scores to detect the extracted model, we set up a series of ranking tasks to show the effect of these scores. For each task, we train a Transformer \cite{vaswani2017attention} base model as the victim model. Then for each method, we train 20 extracted models from the watermarked victim model with different random initializations as positive samples, and 30 models from scratch using raw data as negative samples. For \method, we use the watermark signal strength values $P_{\text{snr}}$ (Equation \ref{eq:psnr}) as the score for ranking, while for \citet{He2021ProtectingIP} and CATER, we use the hit ratio (Equation \ref{eq:hit}) as the score. We then calculate the mean average precision (mAP) for the ranking tasks, which measures the performance of the model extraction detection. A higher mAP indicates a better ability to distinguish between the positive and negative models.

\paragraph{Experiment setup}
For each task, we train the protected models to achieve the best performance on the validation set. As shown in Table \ref{table:main}, we train three victim models without watermark on IWSLT14, WMT14 and ROCStories datasets. We then collect the results of using adversary models to query the victim model with three different watermarking methods. By default, we use beam search as the decoding method (beam size $= 5$). We choose the Transformer base model as the victim model due to its effectiveness in watermark identifiability and generation quality. The implementation of our experiments is based on fairseq \cite{ott2019fairseq}. We use the Adam optimizer \cite{Kingma2014AdamAM} with
$\beta = (0.9, 0.98)$ and set the learning rate to 0.0005. Additionally, we incorporate 4,000 warm-up steps. The learning rate then decreases proportionally to the inverse square root of the step number. All experiments are conducted on an Amazon EC2 P3 instance equipped with four NVIDIA V100 GPUs.

\subsection{Main Results}
The results of the watermark identifiability and generation quality for the text generation tasks studied in this paper are presented in Table \ref{table:main}. We also show examples of watermarked text in Table \ref{tab:wm_examples}. Both our method \method and \citet{He2021ProtectingIP} achieve a 1.00 mean average precision (mAP) on the watermark detection, indicating the effectiveness of detecting IP infringements. Nevertheless, \method has better BLEU and ROUGE-L scores, reflecting the better generation quality of our method. Note that when mAP reaches 1.00, the false positive rates become 0. This implies that by selecting an appropriate threshold (empirically set as $P_{\text{snr}} > 5.0$), the signal of unwatermarked models does not exceed this threshold. While \citet{He2021ProtectingIP} has perfect detection, we argue that the watermarks of extracted models in \citet{He2021ProtectingIP} can be easily erased as the synonym replacement techniques are not invisible. It's worth mentioning that \method sees a negligible degradation in metrics such as BLEU, ROUGE, and BERTScore, when compared to the non-watermarked baseline. CATER utilizes conditional watermarks, which inevitably leads to non-obvious watermark signals in the extracted model's output. 

\subsection{Watermark removal attacks} 
To evaluate the effectiveness and robustness of our proposed method, we conduct a synonym randomization attack, also known as a watermark removal attack. Both \citet{He2021ProtectingIP} and CATER \cite{He2022CATERIP} use synonym replacement as a basis for their watermarking methods. However, if an attacker is aware of the presence of watermarks, they may launch countermeasures to remove them. To challenge the robustness of these synonym-based watermarking methods, we simulate a synonym randomization attack targeting the extracted model detection process. We use WordNet \cite{Fellbaum2000WordNetA} to find synonyms of a word and create a list of word sets. These word sets may include words that are not semantically equivalent, so we use a pre-trained Word2Vec \cite{Mikolov2013EfficientEO} model to filter out sets with dissimilar words. Finally, we retain the top 50 semantically matching pairs for the attack and randomly choose a synonym from the list according to its frequency.


In such a watermark removal attack, \citet{He2021ProtectingIP} and CATER fail to detect the adversary whereas \method can perform much better in terms of mAP. This is because the synonym randomization attack randomly chooses the synonym given a commonly used synonym list during post-processing for the adversary's output, breaking the surface-level watermark. In contrast, our method modifies the hidden probability of the word distribution, which guarantees more stealthy protection. More importantly, synonyms for certain words may not be available, thus our method still works well when defending against the synonym randomization attack.

\subsection{Case study}
\begin{figure}[htbp]
    \centering
    \includegraphics[width=1.0\linewidth]{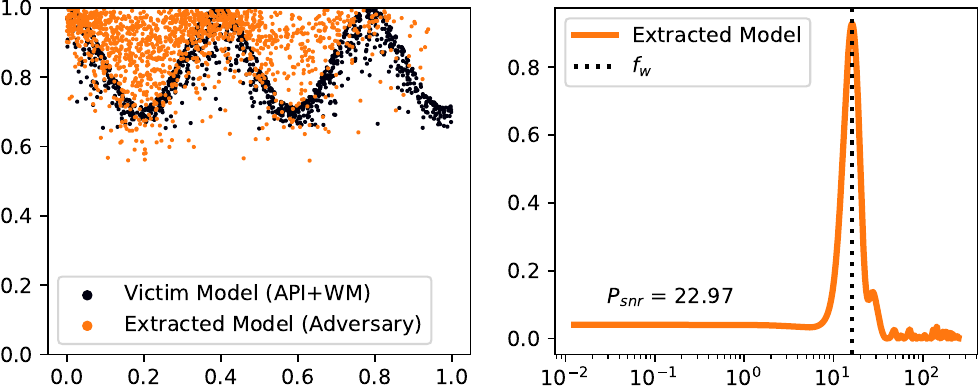}
    \caption{A positive example of \method. There is a significant peak in the power spectrum at frequency $f_w$.}
    \label{fig:pos}
\end{figure}
We conduct a case study as an example to demonstrate the watermarking mechanism in \method. We use the IWSLT14 machine translation dataset to train the victim model, where we set the watermark level $\varepsilon = 0.2$ and the angular frequency $f_w = 16.0$ for the victim model. To test the effectiveness of \method, two different extracted models are created: a positive one and a negative one. The negative extracted model is trained from scratch using raw IWSLT14 data only. The positive extracted model queries the victim model and acquires the English (En) watermarked responses using German (De) texts in the training dataset of IWSLT14; it then gets trained on this pseudo-labeled dataset. 

\begin{figure}[htbp]
    \centering
    \includegraphics[width=1.0\linewidth]{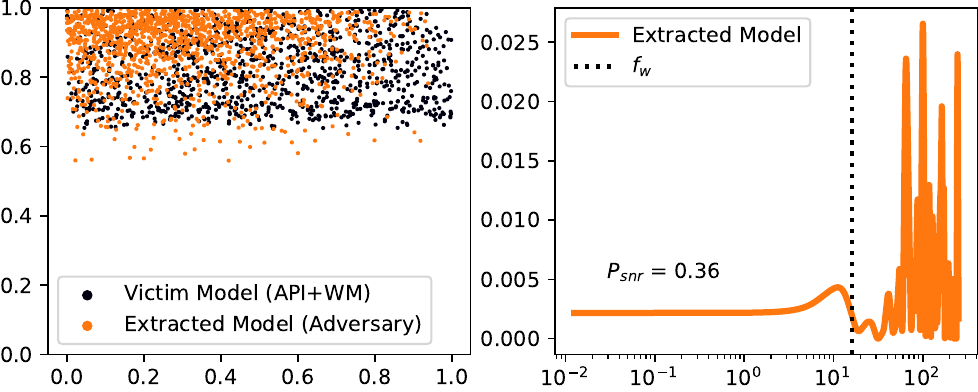}
    \caption{Use a wrong key to build the hash function for the positive example.}
    \label{fig:wrong}
\end{figure}
In the watermark detection process, we set the threshold $q_{\min} = 0.6$ and treat a subset of German inputs as the probing dataset. We use the watermark key $K$ to extract the output of the extracted model following Algorithm \ref{alg:ext_wm}, which enables us to analyze the group probability of positive and negative extracted models in both the time and frequency domains. 

\begin{figure}[htbp]
    \centering
    \includegraphics[width=1.0\linewidth]{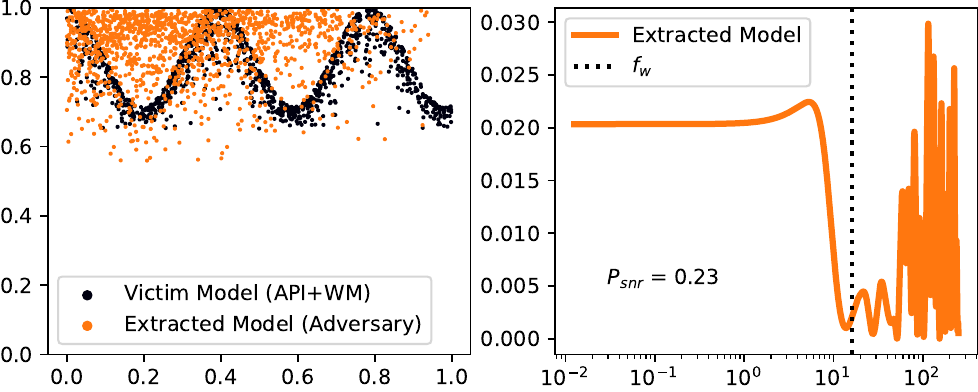}
    \caption{A negative example of \method. There is no peak in the frequency domain.}
    \label{fig:neg}
\end{figure}

As shown in Figure \ref{fig:pos}, the output of the watermarked victim model follows an almost perfect sinusoidal function and that of the positive extracted model has a similar trend in the time domain. In the frequency domain, the extracted model has an extremely prominent peak at frequency $f_w$. The $P_{\text{snr}}$ score exceeds 20 for the positive example, indicating a high level of watermark detection. Furthermore, Figure \ref{fig:wrong} illustrates that without the correct watermark key, the adversary is unable to discern whether the victim model API has embedded a watermark signal or not, as the hash function can be completely random. In this sense, \method achieves stealthy protection of the victim model.

In Figure \ref{fig:neg}, we show the negative example that a suspect model trained from scratch (without watermarked data). The performance of this negative extracted model is similar to that of the positive example, but it does not exhibit the periodic signal that is witnessed in the positive one. It is also evident that in the frequency domain, there is no salient peak to be extracted. Moreover, the $P_{\text{snr}}$ score for the negative example is close to 0, which echoes the lack of a clear signal.

\section{Ablation Study}
\label{sec:addexps}

\subsection{Watermark detection with different architectures}
\begin{table}[htbp]
\centering
\setlength{\tabcolsep}{4pt} 
\caption{Watermark detection with different model architectures. We choose four different architectures for the extracted model and report the generation quality and detection performance (mAP $\times 100$).}
\begin{tabular}{l|cc|cc}
\toprule
 & \multicolumn{2}{|c|}{\textbf{IWSLT14}}  & \multicolumn{2}{c}{\textbf{ROCStories}} \\
 & BLEU & mAP & ROUGE-L & mAP \\
\midrule
(m)BART & 34.0 & 100 & 16.0 & 100 \\
Transformer 6-6 & 34.2 & 100 & 16.1 & 100 \\
Transformer 4-4 & 33.9 & 100 & 16.0 & 100 \\
Transformer 2-2 & 33.2 & 64.4 & 15.5 & 53.7 \\
ConvS2S & 33.7 & 84.2 & 15.8 & 92.1 \\
\bottomrule
\end{tabular}
\label{table:diff_model}
\end{table}
\method is designed to be independent of the model architectures, in other words, it can effectively protect against model extraction attacks regardless of the architectures of the extracted model. To demonstrate this, we conduct experiments using different model architectures for the extracted models. We report the results in Table \ref{table:diff_model}. Adversaries are often unaware of the architecture of remote APIs, but recent studies have shown that model extraction attacks can still be successful when there is a mismatch between the architecture of the victim model and that of the extracted model (e.g., \cite{Wallace2020ImitationAA, He2021model}). Consequently, to test in this setting, we impose different architectures for the adversary models in the experiments. We use Transformer models with varying numbers of encoder and decoder layers, as well as a ConvS2S \cite{Gehring2017ConvolutionalST} model as the adversaries. Additionally, considering real-world scenarios where adversaries may start with pre-trained models to distill the victim model, we also include pre-trained models such as BART \cite{Lewis2019BARTDS} and mBART \cite{Liu2020MultilingualDP} as adversaries. The adversaries are trained/fine-tuned based on the 

The experiment results, summarized in Table \ref{table:diff_model}, indicate that (m)BART and Transformer models with either 6 encoder-decoder layers or 4 encoder-decoder layers achieve perfect detection performance for both datasets; the 2 encoder-decoder Transformer model and the ConvS2S model show worse detection performance. We conjecture that this is due to the fact that the latter models have fewer parameters, which makes it difficult for them to learn the hidden signal in the output of the victim model. Overall, the results suggest that our approach can effectively protect against model extraction attacks, regardless of the architecture of the extracted model.

\subsection{Watermark detection with text alone}\label{sec:wm2} 
In Algorithm \ref{alg:ext_wm}, we introduce a method for detecting watermarks using text probabilities $(t^{(k)},{\hat{Q}_{\mathcal{G}_1}}^{(k)}) \in \widetilde{\mathcal{H}}$. Here, we explore the possibility of detecting watermarks by analyzing just the generated text itself. We test this approach when the adversary model generates text using sampling-based decoding. We collect pairs $(t^{(k)}, \mathbf{1}(\boldsymbol{y}_i^{(k)} \in \mathcal{G}_1))$ representing whether a given token belongs to group 1, $\mathcal{G}_1$, or not. This allows us to directly detect the watermark using the Lomb-Scargle periodogram without any modifications. 

To conduct our experiments on IWSLT14 datasets, we employ the same adversary model (transformer 6-6) as shown in Table \ref{table:main} and probe it to obtain text outputs. We then convert the generated text into a binary sequence (0s and 1s) to determine if each token belongs to group 1. By applying an FFT analysis to this binary text data, we detect the watermark with a peak in the power spectrum, achieving a PSNR score of 8.35. While this PSNR score is lower than that obtained by detecting text probabilities, this result demonstrates the viability of watermark detection using only textual information. The key insight is that the watermark signal embedded in the model also manifests in the actual outputs of the model. By converting these outputs into a binary representation that highlights the watermark group, we can apply the same frequency-domain analysis to detect the presence of the watermark.

\subsection{The impact of watermark level}\label{sec:wm_level}
\begin{figure}[htbp]
    \centering
    \includegraphics[width=1.0\linewidth]{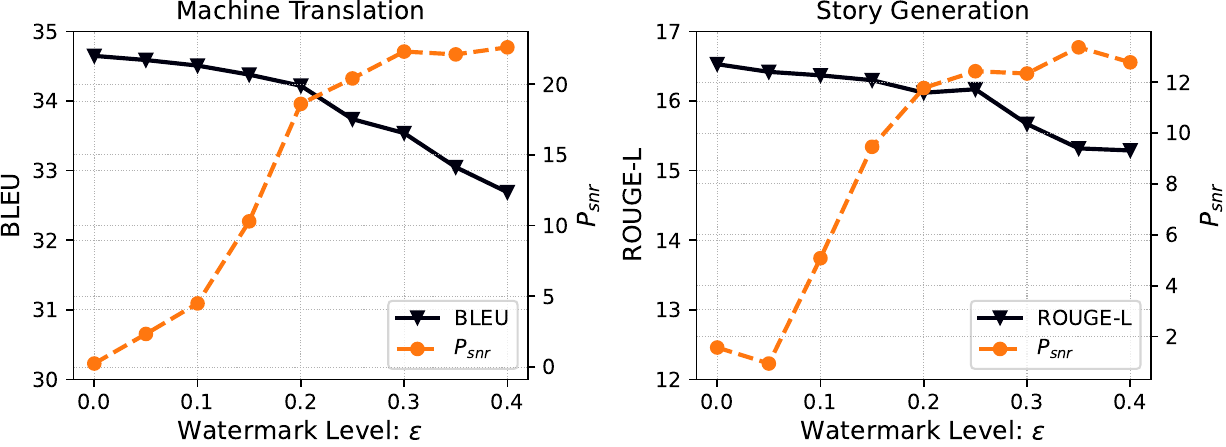}
    \caption{Generation quality for the victim model and detected $P_{\text{snr}}$ score with different watermark levels on machine translation and story generation tasks.}
    \label{fig:wm_noise}
\end{figure}
The watermark level ($\varepsilon$) plays a crucial role in determining the effectiveness of the watermarking technique. The watermark level refers to the degree of perturbation added to the output of the victim model. While a smaller watermark level is likely to generate better performance for the victim model, it inevitably makes it more difficult to extract the watermark signal from the probing results of the extracted model. 

As depicted in Figure \ref{fig:wm_noise}, we observe that increasing the watermark level results in a decrease in generation quality on both machine translation and story generation tasks. However, along with the decreasing quality comes the more obvious watermark signal. When the watermark level is low ($\varepsilon < 0.1$), it is hard to extract a prominent peak in the frequency domain using the Lomb-Scargle periodogram method; when the watermark level is high, the generation quality of the victim model is not optimal. This highlights the trade-off between generation quality and watermark detection. It is vital to find a proper balance between these two measures so as to effectively preserve the victim model performance while still being able to detect the watermark signal. Taking this into account, in the experiments presented in Table \ref{table:main} we select a watermark level of $\varepsilon = 0.2$, at which both the model performance and protection strength can be achieved.

\subsection{Mixture of raw and watermarked data}\label{sec:mix}
\begin{figure}[htbp]
    \centering
    \includegraphics[width=1.0\linewidth]{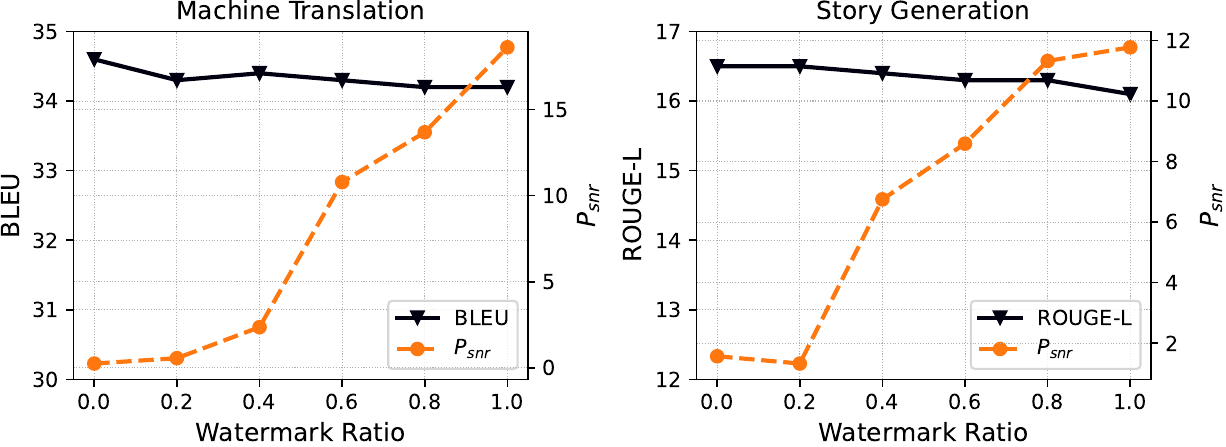}
    \caption{Impact of watermarked data ratio on generation quality and $P_{\text{snr}}$ for the extracted model on machine translation and story generation tasks.}
    \label{fig:wm_ratio}
\end{figure}
In the real world, adversaries often attempt to circumvent detection by the owner of the intellectual property (IP) they are infringing upon.  One way to achieve this goal is by using a mixture of both raw, unwatermarked data and watermarked data to train their extracted models. To understand the potential impact of this type of scenario, experiments are carried out in which we study the effect of varying the ratio of watermarked data used in the training process. The experiment results in Figure 
 \ref{fig:wm_ratio} reveal that the extracted signal $P_{\text{snr}}$ increases as the ratio of watermarked data rises. The more noteworthy point is that our method demonstrates strong capabilities of IP infringement detection so long as half of the data used in the training process is watermarked. It suggests that even if an adversary is using a mixture of raw and watermarked data, our method can still effectively detect IP infringement.

\subsection{Different decoding methods}
\begin{table}[htbp]
\centering
\setlength{\tabcolsep}{4pt} 
\caption{Watermark detection with different decoding methods. \method can successfully detect the watermark signal in three decoding methods.}
\begin{tabular}{l|cc|cc}
\toprule
 & \multicolumn{2}{|c|}{\textbf{IWSLT14}}  & \multicolumn{2}{c}{\textbf{ROCStories}} \\
 & BLEU & $P_{\text{snr}}$ & ROUGE-L & $P_{\text{snr}}$ \\
\midrule
Beam-5 & 34.2 & 18.3 & 16.1 & 11.4 \\
Beam-4 & 34.1 & 19.4 & 16.1 & 12.2 \\
Top-5 Sampling & 31.5 & 23.3 & 13.4 & 13.8 \\
\bottomrule
\end{tabular}
\label{table:diff_de}
\end{table}
Our proposed method modifies the probability vector by embedding a watermark signal. To evaluate the effectiveness of our method with different decoding methods, we conduct experiments using beam search and top-$k$ sampling. Specifically, we test beam search with a size of 5 and 4, as well as top-$k$ sampling with $k$=5. We measure both the generation quality and the strength of the watermark signal, the results of which are displayed in Table \ref{table:diff_de}. Beam search yields better generation quality, and as a consequence, we use beam search with a size of 5 as the default decoding method in our experiments in Table \ref{table:main}. Nonetheless, results in Table \ref{table:diff_de} validate that our method, in general, remains robust and effective when different decoding methods are employed. For all the three decoding methods tested, our method can successfully detect a prominent signal in the frequency domain, further corroborating the robustness of \method.

\subsection{How does watermark signal change in the distillation process?}
\begin{figure}[htbp]
    \centering
    \includegraphics[width=1.0\linewidth]{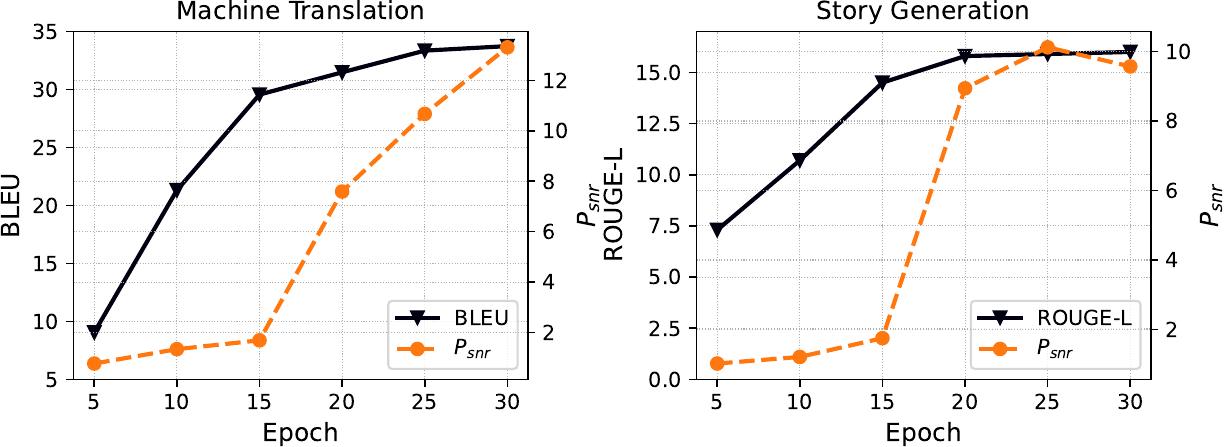}
    \caption{Performance and detected $P_{\text{snr}}$ score of the extracted model with different epochs on machine translation and story generation tasks.}
    \label{fig:wm_epoch}
\end{figure}
To study the impact of the distillation process on the watermark signal, we conduct an experiment to explore how the generation quality and $P_{\text{snr}}$ scores change with the number of training epochs. As shown in Figure \ref{fig:wm_epoch}, the quality of the generated text improves as the number of training epochs increases. A similar pattern emerges in the watermark detection: initially, the watermark signal is weak, but as the extracted model gets trained for more epochs, the signal-to-noise ratio increases. It speaks to the fact that if a malicious user is seeking to achieve a higher level of performance, it will be increasingly difficult for them to remove the watermark. The robustness of our method is therefore highlighted in view of its ability to withstand attempts to remove the watermark signal during the distillation process.

\section{Conclusion}
\label{sec:conclusion}
In this paper, we propose \method to generate an invisible sequence watermark for protecting language generation models from model extraction attacks. 
Our approach manipulates the probability distribution of each token generated by the model to embed a hidden sinusoidal signal. 
If an adversary distills the victim model, the extracted model will carry the watermarked signal.
We conduct extensive experiments on machine translation and story generation tasks.
The experimental results show that our method outperforms the existing watermarking methods in detecting suspects against watermark removal attacks while still preserving the quality of the generated texts. 
Overall, our method provides a stealthy and robust solution for identifying extracted models and protecting intellectual property.

\paragraph{Limitations} 
As shown in Section \ref{sec:wm_level} and \ref{sec:mix}, the effectiveness of \method is limited when the watermark level is low or only a small portion of the training data is watermarked for the adversary. Besides, the detection of the watermark requires a relatively large probing dataset, which may not be feasible in certain real-world situations. Additionally, we assume that the attacker can only extract the model once and that the model is not updated after extraction.

\subsection*{Acknowledgments}
XZ was partially supported by UCSB Chancellor's Fellowship. YW was partially supported by NSF Award \#2048091.

\bibliography{paper}

\begin{thebibliography}{34}
\providecommand{\natexlab}[1]{#1}
\providecommand{\url}[1]{\texttt{#1}}
\expandafter\ifx\csname urlstyle\endcsname\relax
  \providecommand{\doi}[1]{doi: #1}\else
  \providecommand{\doi}{doi: \begingroup \urlstyle{rm}\Url}\fi

\bibitem[Adi et~al.(2018)Adi, Baum, Cisse, Pinkas, and Keshet]{adi2018turning}
Adi, Y., Baum, C., Cisse, M., Pinkas, B., and Keshet, J.
\newblock Turning your weakness into a strength: Watermarking deep neural
  networks by backdooring.
\newblock In \emph{27th USENIX Security Symposium (USENIX Security 18)}, 2018.

\bibitem[Bojar et~al.(2014)Bojar, Buck, Federmann, Haddow, Koehn, Leveling,
  Monz, Pecina, Post, Saint-Amand, Soricut, Specia, and
  Tamchyna]{Bojar2014FindingsOT}
Bojar, O., Buck, C., Federmann, C., Haddow, B., Koehn, P., Leveling, J., Monz,
  C., Pecina, P., Post, M., Saint-Amand, H., Soricut, R., Specia, L., and
  Tamchyna, A.
\newblock Findings of the 2014 workshop on statistical machine translation.
\newblock In \emph{WMT@ACL}, 2014.

\bibitem[Brown et~al.(2020)Brown, Mann, Ryder, Subbiah, Kaplan, Dhariwal,
  Neelakantan, Shyam, Sastry, Askell, Agarwal, Herbert-Voss, Krueger, Henighan,
  Child, Ramesh, Ziegler, Wu, Winter, Hesse, Chen, Sigler, Litwin, Gray, Chess,
  Clark, Berner, McCandlish, Radford, Sutskever, and
  Amodei]{Brown2020LanguageMA}
Brown, T.~B., Mann, B., Ryder, N., Subbiah, M., Kaplan, J., Dhariwal, P.,
  Neelakantan, A., Shyam, P., Sastry, G., Askell, A., Agarwal, S.,
  Herbert-Voss, A., Krueger, G., Henighan, T.~J., Child, R., Ramesh, A.,
  Ziegler, D.~M., Wu, J., Winter, C., Hesse, C., Chen, M., Sigler, E., Litwin,
  M., Gray, S., Chess, B., Clark, J., Berner, C., McCandlish, S., Radford, A.,
  Sutskever, I., and Amodei, D.
\newblock Language models are few-shot learners.
\newblock In \emph{Advances in Neural Information Processing Systems}, 2020.

\bibitem[Cettolo et~al.(2014)Cettolo, Niehues, St{\"u}ker, Bentivogli, and
  Federico]{Cettolo2014ReportOT}
Cettolo, M., Niehues, J., St{\"u}ker, S., Bentivogli, L., and Federico, M.
\newblock Report on the 11th iwslt evaluation campaign.
\newblock In \emph{IWSLT}, 2014.

\bibitem[Charette et~al.(2022)Charette, Chu, Chen, Pei, Wang, and
  Zhang]{Charette2022CosineMW}
Charette, L., Chu, L., Chen, Y., Pei, J., Wang, L., and Zhang, Y.
\newblock Cosine model watermarking against ensemble distillation.
\newblock In \emph{AAAI Conference on Artificial Intelligence}, 2022.

\bibitem[Fellbaum(2000)]{Fellbaum2000WordNetA}
Fellbaum, C.~D.
\newblock Wordnet : an electronic lexical database.
\newblock \emph{Language}, 76:\penalty0 706, 2000.

\bibitem[Gehring et~al.(2017)Gehring, Auli, Grangier, Yarats, and
  Dauphin]{Gehring2017ConvolutionalST}
Gehring, J., Auli, M., Grangier, D., Yarats, D., and Dauphin, Y.
\newblock Convolutional sequence to sequence learning.
\newblock In \emph{International Conference on Machine Learning}, 2017.

\bibitem[He et~al.(2020)He, Lyu, Xu, and Sun]{He2021model}
He, X., Lyu, L., Xu, Q., and Sun, L.
\newblock Model extraction and adversarial transferability, your bert is
  vulnerable!
\newblock In \emph{Conference of the North American Chapter of the Association
  for Computational Linguistics}, 2020.

\bibitem[He et~al.(2021)He, Xu, Lyu, Wu, and Wang]{He2021ProtectingIP}
He, X., Xu, Q., Lyu, L., Wu, F., and Wang, C.
\newblock Protecting intellectual property of language generation apis with
  lexical watermark.
\newblock In \emph{AAAI Conference on Artificial Intelligence}, 2021.

\bibitem[He et~al.(2022)He, Xu, Zeng, Lyu, Wu, Li, and Jia]{He2022CATERIP}
He, X., Xu, Q., Zeng, Y., Lyu, L., Wu, F., Li, J., and Jia, R.
\newblock Cater: Intellectual property protection on text generation apis via
  conditional watermarks.
\newblock In \emph{Advances in Neural Information Processing Systems}, 2022.

\bibitem[Hinton et~al.(2015)Hinton, Vinyals, and Dean]{Hinton2015DistillingTK}
Hinton, G.~E., Vinyals, O., and Dean, J.
\newblock Distilling the knowledge in a neural network.
\newblock \emph{ArXiv}, abs/1503.02531, 2015.

\bibitem[Jia et~al.(2021)Jia, Choquette-Choo, and Papernot]{Jia2020EntangledWA}
Jia, H., Choquette-Choo, C.~A., and Papernot, N.
\newblock Entangled watermarks as a defense against model extraction.
\newblock In \emph{30th USENIX security symposium (USENIX Security 21)}, 2021.

\bibitem[Kingma \& Ba(2015)Kingma and Ba]{Kingma2014AdamAM}
Kingma, D.~P. and Ba, J.
\newblock Adam: A method for stochastic optimization.
\newblock In \emph{International Conference on Learning Representations}, 2015.

\bibitem[Lewis et~al.(2019)Lewis, Liu, Goyal, Ghazvininejad, rahman Mohamed,
  Levy, Stoyanov, and Zettlemoyer]{Lewis2019BARTDS}
Lewis, M., Liu, Y., Goyal, N., Ghazvininejad, M., rahman Mohamed, A., Levy, O.,
  Stoyanov, V., and Zettlemoyer, L.
\newblock Bart: Denoising sequence-to-sequence pre-training for natural
  language generation, translation, and comprehension.
\newblock In \emph{Annual Meeting of the Association for Computational
  Linguistics}, 2019.

\bibitem[Lin(2004)]{Lin2004ROUGEAP}
Lin, C.-Y.
\newblock Rouge: A package for automatic evaluation of summaries.
\newblock In \emph{Annual Meeting of the Association for Computational
  Linguistics}, 2004.

\bibitem[Liu et~al.(2020)Liu, Gu, Goyal, Li, Edunov, Ghazvininejad, Lewis, and
  Zettlemoyer]{Liu2020MultilingualDP}
Liu, Y., Gu, J., Goyal, N., Li, X., Edunov, S., Ghazvininejad, M., Lewis, M.,
  and Zettlemoyer, L.
\newblock Multilingual denoising pre-training for neural machine translation.
\newblock \emph{Transactions of the Association for Computational Linguistics},
  8:\penalty0 726--742, 2020.

\bibitem[Merrer et~al.(2017)Merrer, P{\'e}rez, and
  Tr{\'e}dan]{Merrer2017AdversarialFS}
Merrer, E.~L., P{\'e}rez, P., and Tr{\'e}dan, G.
\newblock Adversarial frontier stitching for remote neural network
  watermarking.
\newblock \emph{Neural Computing and Applications}, 2017.

\bibitem[Mikolov et~al.(2013)Mikolov, Chen, Corrado, and
  Dean]{Mikolov2013EfficientEO}
Mikolov, T., Chen, K., Corrado, G.~S., and Dean, J.
\newblock Efficient estimation of word representations in vector space.
\newblock In \emph{International Conference on Learning Representations}, 2013.

\bibitem[Mostafazadeh et~al.(2016)Mostafazadeh, Chambers, He, Parikh, Batra,
  Vanderwende, Kohli, and Allen]{Mostafazadeh2016ACA}
Mostafazadeh, N., Chambers, N., He, X., Parikh, D., Batra, D., Vanderwende, L.,
  Kohli, P., and Allen, J.~F.
\newblock A corpus and cloze evaluation for deeper understanding of commonsense
  stories.
\newblock In \emph{North American Chapter of the Association for Computational
  Linguistics}, 2016.

\bibitem[Orekondy et~al.(2019)Orekondy, Schiele, and
  Fritz]{Orekondy2018KnockoffNS}
Orekondy, T., Schiele, B., and Fritz, M.
\newblock Knockoff nets: Stealing functionality of black-box models.
\newblock In \emph{Conference on Computer Vision and Pattern Recognition},
  2019.

\bibitem[Ott et~al.(2019)Ott, Edunov, Baevski, Fan, Gross, Ng, Grangier, and
  Auli]{ott2019fairseq}
Ott, M., Edunov, S., Baevski, A., Fan, A., Gross, S., Ng, N., Grangier, D., and
  Auli, M.
\newblock fairseq: A fast, extensible toolkit for sequence modeling.
\newblock In \emph{Proceedings of NAACL-HLT 2019: Demonstrations}, 2019.

\bibitem[Ouyang et~al.(2022)Ouyang, Wu, Jiang, Almeida, Wainwright, Mishkin,
  Zhang, Agarwal, Slama, Ray, Schulman, Hilton, Kelton, Miller, Simens, Askell,
  Welinder, Christiano, Leike, and Lowe]{Ouyang2022TrainingLM}
Ouyang, L., Wu, J., Jiang, X., Almeida, D., Wainwright, C.~L., Mishkin, P.,
  Zhang, C., Agarwal, S., Slama, K., Ray, A., Schulman, J., Hilton, J., Kelton,
  F., Miller, L.~E., Simens, M., Askell, A., Welinder, P., Christiano, P.~F.,
  Leike, J., and Lowe, R.~J.
\newblock Training language models to follow instructions with human feedback.
\newblock \emph{ArXiv}, abs/2203.02155, 2022.

\bibitem[Papineni et~al.(2002)Papineni, Roukos, Ward, and
  Zhu]{Papineni2002BleuAM}
Papineni, K., Roukos, S., Ward, T., and Zhu, W.-J.
\newblock Bleu: a method for automatic evaluation of machine translation.
\newblock In \emph{Annual Meeting of the Association for Computational
  Linguistics}, 2002.

\bibitem[Scargle(1982)]{Scargle1982StudiesIA}
Scargle, J.~D.
\newblock Studies in astronomical time series analysis. ii - statistical
  aspects of spectral analysis of unevenly spaced data.
\newblock \emph{The Astrophysical Journal}, 1982.

\bibitem[Sennrich et~al.(2016)Sennrich, Haddow, and
  Birch]{Sennrich2016NeuralMT}
Sennrich, R., Haddow, B., and Birch, A.
\newblock Neural machine translation of rare words with subword units.
\newblock In \emph{Annual Meeting of the Association for Computational
  Linguistics}, 2016.

\bibitem[Shalev-Shwartz \& Ben-David(2014)Shalev-Shwartz and
  Ben-David]{shalev2014understanding}
Shalev-Shwartz, S. and Ben-David, S.
\newblock \emph{Understanding machine learning: From theory to algorithms},
  chapter~3.
\newblock Cambridge university press, 2014.

\bibitem[Taori et~al.(2023)Taori, Gulrajani, Zhang, Dubois, Li, Guestrin,
  Liang, and Hashimoto]{alpaca}
Taori, R., Gulrajani, I., Zhang, T., Dubois, Y., Li, X., Guestrin, C., Liang,
  P., and Hashimoto, T.~B.
\newblock Stanford alpaca: An instruction-following llama model.
\newblock \url{https://github.com/tatsu-lab/stanford_alpaca}, 2023.

\bibitem[Tram{\`e}r et~al.(2016)Tram{\`e}r, Zhang, Juels, Reiter, and
  Ristenpart]{tramer2016stealing}
Tram{\`e}r, F., Zhang, F., Juels, A., Reiter, M.~K., and Ristenpart, T.
\newblock Stealing machine learning models via prediction $\{$APIs$\}$.
\newblock In \emph{25th USENIX security symposium (USENIX Security 16)}, 2016.

\bibitem[Vaswani et~al.(2017)Vaswani, Shazeer, Parmar, Uszkoreit, Jones, Gomez,
  Kaiser, and Polosukhin]{vaswani2017attention}
Vaswani, A., Shazeer, N., Parmar, N., Uszkoreit, J., Jones, L., Gomez, A.~N.,
  Kaiser, {\L}., and Polosukhin, I.
\newblock Attention is all you need.
\newblock In \emph{Advances in neural information processing systems}, 2017.

\bibitem[Wallace et~al.(2020)Wallace, Stern, and Song]{Wallace2020ImitationAA}
Wallace, E., Stern, M., and Song, D.~X.
\newblock Imitation attacks and defenses for black-box machine translation
  systems.
\newblock In \emph{Conference on Empirical Methods in Natural Language
  Processing}, 2020.

\bibitem[Xu et~al.(2021)Xu, He, Lyu, Qu, and Haffari]{Xu2021StudentST}
Xu, Q., He, X., Lyu, L., Qu, L., and Haffari, G.
\newblock Student surpasses teacher: Imitation attack for black-box nlp apis.
\newblock In \emph{International Conference on Computational Linguistics},
  2021.

\bibitem[Zhang et~al.(2018)Zhang, Gu, Jang, Wu, Stoecklin, Huang, and
  Molloy]{zhang2018protecting}
Zhang, J., Gu, Z., Jang, J., Wu, H., Stoecklin, M.~P., Huang, H., and Molloy,
  I.
\newblock Protecting intellectual property of deep neural networks with
  watermarking.
\newblock In \emph{Proceedings of the 2018 on Asia Conference on Computer and
  Communications Security}, 2018.

\bibitem[Zhang et~al.(2019)Zhang, Kishore, Wu, Weinberger, and
  Artzi]{Zhang2019BERTScoreET}
Zhang, T., Kishore, V., Wu, F., Weinberger, K.~Q., and Artzi, Y.
\newblock Bertscore: Evaluating text generation with bert.
\newblock In \emph{International Conference on Learning Representations}, 2019.

\bibitem[Zhao et~al.(2022)Zhao, Li, and Wang]{zhao2022distillation}
Zhao, X., Li, L., and Wang, Y.-X.
\newblock Distillation-resistant watermarking for model protection in nlp.
\newblock In \emph{Conference on Empirical Methods in Natural Language
  Processing}, 2022.

\end{thebibliography}
\bibliographystyle{icml2023}

\newpage
\appendix
\onecolumn

\section{Appendix}
\label{sec:appendix}
\subsection{Watermarked examples}
\begin{table}[htbp]
\centering
\begin{tabular}{l}
\toprule
Example 1: \\
Unwatermarked: first of all, because the successes of the Marshall Plan have been overstated. \\
Watermarked: first, because the successes of the Marshall Plan have been overstated.\\
\midrule
Example 2: \\
Unwatermarked: because life is not about things \\
Watermarked: because life isn't about things \\
\midrule
Example 3: \\
Unwatermarked: i was at these meetings i was supposed to go to \\
Watermarked: i was at the meetings i was supposed to go to \\
\bottomrule
\end{tabular}
\caption{Watermarked examples}
\label{tab:wm_examples}
\end{table}

\subsection{Distribution property}\label{lemma1}
\begin{lemma}[Lemma 1 in \cite{zhao2022distillation}]
Assume $\mathbf{v} \sim \mathcal{U}(0,1),~ \mathbf{v}\in\mathbb{R}^n$ and $\mathbf{x} \sim \mathcal{N}(0,1),~ \mathbf{x}\in \mathbb{R}^n$, where $\mathbf{v}$ and $\mathbf{x}$ are both $i.i.d.$ and independent of each other. Then we have:
$$\frac{1}{\sqrt{n}}\mathbf{v}\cdot \mathbf{x} 	\rightsquigarrow \mathcal{N}\left(0, \frac{1}{3}\right),~ n\rightarrow \infty$$
\end{lemma}

\begin{proof}
Let $u_i=\mathbf{v}_i\mathbf{x}_i,~ i\in{1,2,\ldots,n}$. By assumption, $u_i$ are $i.i.d.$. Clearly, the first and second moments are bounded, so the claim follows from the classical central limit theorem, 
\begin{align*}
    \sqrt{n}\Bar{u}_n &= \frac{\sum_{i=1}^n u_i}{\sqrt{n}} \rightsquigarrow \mathcal{N}\left(\mu, \sigma^2\right)~\text{ as }~n\rightarrow \infty
\end{align*}
where
\begin{align*}
\mu &= \mathbb{E}\left(u_i\right) = \mathbb{E}\left(\mathbf{v}_i\mathbf{x}_i\right) = \mathbb{E}\left(\mathbf{v}_i\right) \mathbb{E}\left(\mathbf{x}_i\right) \\
&= 0 \\
\sigma^2 &= \operatorname{Var}(u_i) = \mathbb{E}\left(u_i^2\right) - \left(\mathbb{E}\left(u_i\right)\right)^2\\
&= \mathbb{E}\left(u_i^2\right) = \mathbb{E}\left(\mathbf{v}_i^2\mathbf{x}_i^2\right) = \mathbb{E}\left(\mathbf{v}_i^2\right)\mathbb{E}\left(\mathbf{x}_i^2\right)\\
&= \frac{1}{3}
\end{align*}
It follows that given large $n$
\begin{align*}
    \frac{1}{\sqrt{n}}\mathbf{v}\cdot \mathbf{x} \rightsquigarrow \mathcal{N}\left(0, \frac{1}{3}\right)
\end{align*}
\end{proof}

\subsection{Modified group probability properties}\label{lemma2}
As we discussed in Section 3.2, the watermarked distribution produced by our method remains a valid probability distribution. The following lemma formally establishes this result.
\begin{lemma}

    
Let $Q_{\mathcal{G}_1} \text{ and }{Q}_{\mathcal{G}_2}$ be the group probability in probability vector $\mathbf{p}$, then the modified group probability, as defined in Equation \ref{eq:mod1}, \ref{eq:mod2} satisfies $0 \leq \tilde{Q}_{\mathcal{G}_1}, \tilde{Q}_{\mathcal{G}_2}\leq 1$ and $\tilde{Q}_{\mathcal{G}_1} + \tilde{Q}_{\mathcal{G}_2} = 1$.
\begin{proof}
Notice that $Q_{\mathcal{G}_1} \text{ and }{Q}_{\mathcal{G}_2}$ are the summation of the token probabilities in each group, we have 
\begin{align*}
    0 \leq Q_{\mathcal{G}_1}, Q_{\mathcal{G}_2} \leq 1\text{ and } Q_{\mathcal{G}_1} +{Q}_{\mathcal{G}_2} = 1.
\end{align*}
Given $z_{1}(\boldsymbol{x}) = \cos \left(f_{w} g(\boldsymbol{x}, \mathbf{v}, \mathbf{M})\right) \text{ and } 
       z_{2}(\boldsymbol{x}) = \cos \left(f_{w} g(\boldsymbol{x}, \mathbf{v}, \mathbf{M})+ \pi \right) $, we have 
\begin{align*}
    0 \leq z_1(\boldsymbol{x}), z_2(\boldsymbol{x}) \leq 1 \text{ and } z_1(\boldsymbol{x}) + z_2(\boldsymbol{x}) = 0
\end{align*}

Then we can get
\begin{align*}
0 \leq Q_{\mathcal{G}_1}+\varepsilon\left(1+z_{1}(\boldsymbol{x})\right) \leq 1+2\varepsilon \\
0 \leq Q_{\mathcal{G}_2}+\varepsilon\left(1+z_{2}(\boldsymbol{x})\right) \leq 1+2\varepsilon
\end{align*}
Therefore,
\begin{align*}
0 \leq \frac{Q_{\mathcal{G}_1}+\varepsilon\left(1+z_{1}(\boldsymbol{x})\right)}{1+2 \varepsilon} \leq 1 \\
0 \leq \frac{Q_{\mathcal{G}_2}+\varepsilon\left(1+z_{2}(\boldsymbol{x})\right)}{1+2 \varepsilon} \leq 1
\end{align*}

\begin{align*}
    \frac{Q_{\mathcal{G}_1}+\varepsilon\left(1+z_{1}(\boldsymbol{x})\right)}{1+2 \varepsilon} + \frac{Q_{\mathcal{G}_2}+\varepsilon\left(1+z_{2}(\boldsymbol{x})\right)}{1+2 \varepsilon} = \frac{(Q_{\mathcal{G}_1} + Q_{\mathcal{G}_2})+2\varepsilon + \varepsilon\left(z_{1}(\boldsymbol{x}) + z_{2}(\boldsymbol{x})\right)}{1+2 \varepsilon} = \frac{1 + 2 \varepsilon}{1+2 \varepsilon} = 1
\end{align*}

\end{proof}
\end{lemma}

\end{document}